\documentclass[letterpaper, 10pt, conference]{ieeeconf}
\IEEEoverridecommandlockouts
\usepackage{graphicx, enumerate}
\usepackage{cite,array, bm, epstopdf}
\usepackage{dsfont}
\usepackage{mathrsfs}
\usepackage{amssymb, amsmath}
\usepackage{subeqnarray}
\usepackage{cases}
\addtocounter{MaxMatrixCols}{20}
\usepackage{amssymb, algorithm, algorithmic}
\usepackage{amsmath}
\usepackage{stfloats}

\newtheorem{assumption}{Assumption}
\newtheorem{theorem}{Theorem}

\newtheorem{lemma}{Lemma}

\newtheorem{definition}{Definition}
\newtheorem{remark}{Remark}

\DeclareMathOperator{\p}{p}
\DeclareMathOperator{\ct}{c}
\DeclareMathOperator{\mati}{mati}
\DeclareMathOperator{\mad}{mad}

\DeclareMathOperator{\diag}{diag}

\DeclareMathOperator{\prb}{\mathds{P}}
\DeclareMathOperator{\E}{\mathds{E}}

\begin{document}

\title{\LARGE \bf Stability and $\mathcal{H}_{\infty}$ Performance Analysis of Stochastic Linear Networked and Quantized Control Systems
\thanks{This work was supported by National Natural Science Foundation of China, under Grant 61773357.}
}

\author{Wei~Ren and Junlin~Xiong, \IEEEmembership{Member, IEEE}
\thanks{W. Ren is with Division of Decision and Control Systems, EECS, KTH Royal Institute of Technology, SE-10044, Stockholm, Sweden. J. Xiong is with Department of Automation, University of Science and Technology of China, Hefei, 230026, Anhui, China.
Email: \texttt{\small gtppwe@gmail.com}, \texttt{\small junlin.xiong@gmail.com}.}
}

\maketitle

\begin{abstract}
This paper studies the stability and $\mathcal{H}_{\infty}$ performance analysis problem for linear networked and quantized control systems with both communication delays random packet losses. To deal with the network-induced uncertainties and random packet dropouts, a novel discrete-time stochastic system model is developed for continuous-time networked control systems, and further overapproximated via a polytopic system with norm-bounded uncertainty. Based on the overapproximated system model, sufficient conditions are established for linear networked and quantized control systems in different cases to guarantee input-to-state stability and $\mathcal{H}_{\infty}$ performance with respect to the network-induced errors. Finally, a numerical example is presented to illustrated the developed results.
\end{abstract}

\section{Introduction}
\label{sec-intro}

Feedback control systems with (part of) control loop closed over a shared wired/wireless communication network are called networked control systems (NCSs) \cite{Lunze2014control}. The presence of the communication network in control systems offers considerable benefits, such as increased flexibility and low installation and maintenance cost. Meanwhile, the capability-limited network also induces many issues, like quantization errors, time-varying transmission intervals, time-varying transmission delays, packet dropouts and communication constraints, which have great effects on different performances of NCSs. As a result, system modelling, stability analysis and controller design of NCSs have attracted numerous attention in the past few decades, and individual or joint issues have been considered in previous works; see \cite{Donkers2011stability, Heemels2010networked, Nesic2009unified, Xia2011analysis, Zhang2013network} and references therein.

Due to the aforementioned issues, there are two main types of system models for NCSs in the literature: the deterministic models \cite{Nesic2004input, Cloosterman2010controller, Donkers2011stability, Walsh2002stability} and the stochastic models \cite{Donkers2012stability, Tabbara2008input, Tsumura2009tradeoffs, Antunes2013stability}. In the deterministic models, the effects of different issues are limited with some hard bounds. For instance, transmission intervals and transmission delays are bounded in some given intervals, the dropout packets are transformed into the constraints on the maximally allowable transmission interval \cite{Heemels2010networked, Van2012discrete}. However, since these hard bounds result in the conservatism of the obtained conditions and some network-induced issues have certain stochastic natures, the stochastic models are proposed to characterize these stochastic natures more accurately and to derive less conservative results. For instance, the packet dropouts are stochastic due to the slotted carrier sense multiple access with collision avoidance (CSMA/CA) mechanism in \cite{Park2009generalized}, and have been modelled as a Bernoulli distributed sequence in \cite{Tsumura2009tradeoffs, Quevedo2011packetized, Yang2006control, Wang2007robust}. Furthermore, both transmission intervals and delays were assumed to be characterized by some probability functions in \cite{Antunes2013stability, Donkers2012stability, Xu2012stochastic}. Even though different stochastic models have been constructed in the literature, only parts of the aforementioned network-induced issues are considered, which motives us to focus on this topic further.

In this paper, we study the stability and $\mathcal{H}$ performance analysis problem of linear networked control systems with all aforementioned network-induced issues, which are called networked and quantized control systems (NQCSs). In particular, transmission intervals, communication delays and packet losses are allowed to be random. To this end, our first contribution is to develop a novel system model for the linear NQCS. Using the discrete-time modelling approach, the considered continuous-time system is first discretised into the discrete-time system with random parameters. Furthermore, following the approximation technique applied in \cite{Donkers2011stability, Donkers2012stability, Van2013stability}, the set of admissible transmission intervals and delays is approximated, and hence random parameters are isolated. Therefore, the discrete-time system is overapproximated by a polytopic system with norm-bounded uncertainties. The developed system is general due to the following reasons: (1) all the network-induced issues are studied in this paper, whereas the quantization has not been investigated in \cite{Donkers2011stability, Donkers2012stability, Cloosterman2010controller}; (2) transmission intervals, communication delays and the packet dropouts are allowed to be random, while the effects of all the network-induced issues were constrained with hard bounds in \cite{Donkers2011stability, Van2013stability, Cloosterman2010controller}.

According to the developed polytopic system, both the exponential mean-square input-to-state stability (EMSISS) and the $\mathcal{H}_{\infty}$ performance are studied, which is the second contribution of this paper. Since the polytopic system is the overapproximation of the original system, the stability and performance analysis is on the polytopic system. For two different time-scheduling protocols, sufficient conditions are derived in terms of linear matrix inequalities (LMIs). Due to the randomness of transmission intervals, communication delays and the packet dropouts, the considered stability property and $\mathcal{H}_{\infty}$ performance are the stochastic version, and different from those in \cite{Donkers2011stability, Donkers2012stability, Van2013stability}. In particular, different from the $\mathcal{H}_{\infty}$ performance with respect to the external disturbance in \cite{Yang2006control, Wang2007robust}, the disturbance here comes from the network-induced issues, and can be treated as the interior disturbance. As a result, the obtained results provide alternative conditions for system stability and performances.

Th reminder of this paper is organized as follows. The discrete-time system model with both stochastic parameters and uncertainties is developed for linear NQCSs in Section \ref{sec-problemformation}. We approximate the discrete stochastic system model as a polytopic system with the norm-bounded uncertainty in Section \ref{sec-overapproximation}. Using the developed approximated system, sufficient conditions are established to guarantee system stability and the $\mathcal{H}_{\infty}$ performance in Section \ref{sec-stabilityanalysis}. A numerical example is presented in Section \ref{sec-example}. Conclusion and the future work are stated in Section \ref{sec-conclusion}.

\emph{Notation:} $\mathbb{R}:=(-\infty, +\infty)$; $\mathbb{R}^{+}:=[0, +\infty)$; $\mathbb{N}:=\{0, 1, \ldots\}$; $\mathbb{N}^{+}:=\{1, 2, \ldots\}$. $\|\cdot\|$ stands for Euclidean norm. $\mathbb{B}$ stands for the unit ball and $\mathbb{B}^{\circ}$ denotes the interior of the unit ball; $\mathds{B}(a, b)$ denotes the hypertopic box centered at $a\in\mathbb{R}^{n}$ with the edges $2b$. Given a square matrix $A$, the superscript `$\top$' stands for matrix transposition; the symbol `$\star$' denotes the symmetric part for a symmetric matrix; $\lambda_{\min}(A)$ and $\lambda_{\max}(A)$ are the minimal and maximal eigenvalues of $A$, respectively. $\mathds{P}\{\cdot\}$ denotes the probability and $\mathds{E}\{\cdot\}$ denotes the expectation. To shorten notation, $(x^{\top}, y^{\top})^{\top}$ is denoted by $(x, y)$.

\section{System Model and Problem Formulation}
\label{sec-problemformation}

In this section, a discrete-time stochastic system model is developed for linear NQCSs, the configuration of which is presented in Fig. \ref{fig-1}. Consider the following continuous-time linear system:
\begin{equation}
\label{eqn-1}
\dot{x}_{\p}(t)=A_{\p}x_{\p}(t)+B_{\p}\hat{u}(t), \quad y(t)=C_{\p}x_{\p}(t),
\end{equation}
where $x_{\p}\in\mathbb{R}^{n_{\p}}$ is the system state, $\hat{u}\in\mathbb{R}^{n_{u}}$ is the most recent control input available at the actuator, and $y\in\mathbb{R}^{n_{y}}$ is the system output. The continuous-time linear controller designed for the system \eqref{eqn-1} is of the form:
\begin{equation}
\label{eqn-2}
\dot{x}_{\ct}(t)=A_{\ct}x_{\ct}(t)+B_{\ct}\hat{y}(t), \quad u(t)=C_{\ct}x_{\ct}(t)+D_{\ct}\hat{y}(t),
\end{equation}
where $x_{\ct}\in\mathbb{R}^{n_{\ct}}$ is the controller state, $\hat{y}\in\mathbb{R}^{n_{y}}$ is the most recent plant output available at the controller and $u\in\mathbb{R}^{n_{u}}$ is the controller output.

\subsection{Information Transmission}

\begin{figure}[!t]
\begin{center}
\begin{picture}(65, 95)
\put(-65, -12){\resizebox{65mm}{35mm}{\includegraphics[width=2.5in]{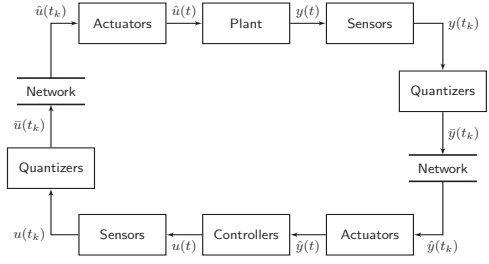}}}
\end{picture}
\end{center}
\caption{General framework of networked and quantized control systems.}
\label{fig-1}
\end{figure}

The outputs of the plant \eqref{eqn-1} and the controller \eqref{eqn-2} are sampled, quantized and then transmitted via the network.
According to the number of sensors and actuators, the network is assumed to have $L\in\mathbb{N}^{+}$ nodes. Correspondingly, the outputs of the plant and the controller are partitioned into $L$ parts. At each transmission time, one and only one node is allowed to transmit the information via the network. Which node is granted to access to the network depends on the time-scheduling protocol, which will be given in Subsection \ref{sec-protocols}. The sampling time sequence is denoted by $\{t_{k}\in\mathbb{R}^{+}: k\in\mathbb{N}\}$, which is aperiodic and strictly increasing, and thus the transmission intervals are defined as $h_{k}:=t_{k+1}-t_{k}$, $k\in\mathbb{N}$. Because of the measurement computation and transmission among digital equipments (e.g., the network and quantizer), there is a sequence of the transmission delays $\tau_{k}\geq0$, $k\in\mathbb{N}^{+}$. Assume that the transmission intervals and the transmission delays are random variables and described by an independent and identically distributed (i.i.d.) stochastic process. To be specific, the following assumption is made for the transmission intervals and the transmission delays.

\begin{assumption}
\label{asn-1}
For each $k\in\mathbb{N}^{+}$, the transmission interval $h_{k}$ and the transmission delay $\tau_{k}$ are i.i.d. random variables, which are characterized by a probability distribution with $\prb\{(h, \tau)\in(\Theta+\delta\mathbb{B}^{\circ})\}=1$, where $\Theta=\{(h, \tau)\in\mathbb{R}^{2}: h\in[\varepsilon, h_{\mati}], \tau\in[0, \min\{h, \tau_{\mad}\})\}$, $\varepsilon\in(0, h_{\mati})$, $\delta>0$ and $h_{\mati}\geq\tau_{\mad}\geq0$.
\end{assumption}

In Assumption \ref{asn-1}, $\varepsilon>0$ is used to avoid Zeno solutions and comes from hardware constraints of the network and quantizer. The constants $h_{\mati}$ and $h_{\mad}$ are respectively called the maximally allowable transfer interval (MATI) and the maximally allowable delay (MAD), both of which are the parameter to be designed.

\begin{remark}
\label{rmk-1}
Assumption \ref{asn-1} recovers many scenarios in the literature \cite{Donkers2012stability, Heemels2009networked, Van2013stability} as the special cases. For instance, the set $\Theta$ is reduced to be $\{(h, \tau)\in\mathbb{R}^{2}: h>0, \tau\in[0, h)\}$ for the stochastic networked control systems \cite{Donkers2012stability}. For the deterministic networked and quantized control systems \cite{Heemels2009networked, Van2013stability}, Assumption \ref{asn-1} is reduced to be deterministic.
\hfill $\square$
\end{remark}

After sampling, the sampled information is quantized to match the limited capacity of the network. Each node has a quantizer, which is a piecewise continuous function $\mathbf{q}_{j}: \mathbb{R}^{n_{j}}\rightarrow\mathcal{Q}_{j}\subset\mathbb{R}^{n_{j}}$, where $\mathcal{Q}_{j}$ is a finite or countable set and $j\in\mathfrak{L}$. The quantizer is assumed to  satisfy the following assumption.

\begin{assumption}[\cite{Ren2018quantized, Liberzon2003hybrid, Liberzon2007input}]
\label{asn-2}
For each $j\in\mathfrak{L}$, there exist $M_{j}>\Lambda_{j}>0$ and $\Lambda_{0j}>0$ such that: (i) $\|z_{j}\|\leq M_{j}\Rightarrow\|\mathbf{q}(z_{j})-z_{j}\|\leq\Lambda_{j}$; (ii) $\|z_{j}\|>M_{j}\Rightarrow\|\mathbf{q}_{j}(z_{j})\|>M_{j}-\Lambda_{j}$; (iii) $\|z_{j}\|\leq\Lambda_{0j}\Rightarrow\mathbf{q}_{j}(z_{j})\equiv0$.
\end{assumption}

Based on Assumption \ref{asn-2}, the applied quantizer is given by
\begin{equation}
\label{eqn-3}
q_{j}(\mu_{j}, z_{j}):=\mu_{j}\mathbf{q}(z_{j}/\mu_{j}),
\end{equation}
where $\mu_{j}\in\mathbb{R}_{+}$ is the time-varying quantization parameter. For the quantizer \eqref{eqn-3}, the range is $M_{j}\mu_{j}$ and the upper bound of the quantization error is $\Lambda_{j}\mu_{j}$, both of which depend on $\mu_{j}\in\mathbb{R}_{+}$ and thus are dynamic. For the quantizer \eqref{eqn-3}, the following assumption is satisfied; see \cite{Heemels2009networked, Nesic2009unified}.

\begin{assumption}
\label{asn-3}
The bound on the initial state $x_{\p}(t_{0})\in\mathbb{R}^{n_{\p}}$ is known \textit{a priori}. Both $x_{\p}(t_{0})$ and $\mu$ are such that the system state is in the quantization regions.
\end{assumption}

All the quantization parameters in $L$ nodes are combined as $\mu:=(\mu_{1}, \ldots, \mu_{L})\in\mathbb{R}^{L}_{+}$, and then the overall quantizer is defined as $q(\mu, z):=(q_{1}(\mu_{1}, z_{1}), \ldots, q_{L}(\mu_{L}, z_{L}))$. The evolution of $\mu\in\mathbb{R}^{L}_{+}$ will be given later. The quantized measurements are defined as $\bar{y}:=q(\mu, y)\in\mathbb{R}^{n_{y}}$ and $\bar{u}:=q(\mu, u)\in\mathbb{R}^{n_{u}}$, and thus the quantization errors are defined as $\epsilon_{y}:=\bar{y}-y\in\mathbb{R}^{n_{y}}$ and $\epsilon_{u}:=\bar{u}-u\in\mathbb{R}^{n_{u}}$.

The quantized measurements are transmitted via the network. The received information is defined as $\hat{y}$ and $\hat{u}$, and updated with the latest quantized measurements at the update times $r_{k}:=t_{k}+\tau_{k}$, $k\in\mathbb{N}^{+}$. Therefore, the received information is updated with the latest quantized measurements. That is, the chosen node to transmit the data packet is updated, and the unchosen nodes are kept the same as before. Here, assume that the packet dropouts may occur randomly; see \cite{Yang2006control, Wang2007robust, Tsumura2009tradeoffs}. For the updates of $\hat{y}$ and $\hat{u}$, two random variables are introduced to model the transmission status. $\alpha_{k}$ and $\beta_{k}$ are independent Bernoulli distributed white sequences with
\begin{align}
\label{eqn-4}
\begin{aligned}
\prb\{\alpha_{k}=1\}&=\E\{\alpha_{k}\}=\bar{\alpha}, \quad \prb\{\alpha_{k}=0\}=1-\bar{\alpha}, \\
\prb\{\beta_{k}=1\}&=\E\{\beta_{k}\}=\bar{\beta}, \quad \prb\{\beta_{k}=0\}=1-\bar{\beta},
\end{aligned}
\end{align}
where $\bar{\alpha}, \bar{\beta}\in(0, 1)$ are called the dropout rates from sensors to the controller and from the controller to actuators, respectively. Hence, at the arrival time $r_{k}\in\mathbb{R}_{+}$, $k\in\mathbb{N}_{+}$, the received information is updated as follows:
\begin{align}
\label{eqn-5}
\begin{aligned}
\hat{y}(r^{+}_{k})&=\alpha_{k}\Gamma^{y}_{\sigma_{k}}(y(t_{k})+\epsilon_{y}(t_{k}))+(I-\alpha_{k}\Gamma^{y}_{\sigma_{k}})\hat{y}(t_{k}),\\
\hat{u}(r^{+}_{k})&=\beta_{k}\Gamma^{u}_{\sigma_{k}}(u(t_{k})+\epsilon_{u}(t_{k}))+(I-\beta_{k}\Gamma^{u}_{\sigma_{k}})\hat{u}(t_{k}),
\end{aligned}
\end{align}
where $\Gamma_{y}(\sigma_{k})$ and $\Gamma_{u}(\sigma_{k})$ are diagonal matrices with the form $\diag\{\delta_{i, 1}, \ldots, \delta_{i, n}\}$, and $\sigma_{k}\in\mathfrak{L}:=\{1, \ldots, L\}$. $\delta_{i,j}=1$ if the node $j$ is granted to access to the network at $t_{i}$; otherwise, $\delta_{i, j}\equiv0$. In the update intervals, the received information is operated in zero order hold (ZOH) fashion. The errors induced by the network and the quantizer are defined as $e_{y}(t)=\hat{y}(t)-y(t)$, $e_{u}(t)=\hat{u}(t)-u(t)$.

Assume that the quantized measurements are operated in ZOH fashion on $(r_{i}, r_{i+1})$ and updated at $r_{i}$ with the update of $\mu\in\mathbb{R}^{L}_{+}$. The update of $\mu$ is given by
\begin{align*}
\mu(r^{+}_{k})&=(1-\min\{\alpha_{k}, \beta_{k}\})\mu(r_{k})+\min\{\alpha_{k}, \beta_{k}\}\Omega\mu(r_{k}),
\end{align*}
where $\Omega:=\diag\{\Omega_{1}, \ldots, \Omega_{L}\}$ with $\Omega_{j}\in(0, 1)$ for all $j\in\mathfrak{L}$. Therefore, $\mu$ is updated only when the information is transmitted such that the whole system is closed-loop, which is different from the deterministic case \cite{Van2013stability, Nesic2009unified, Elia2001stabilization}.

\subsection{Development of Discrete-time System Model}
\label{sec-problem}

Using discrete-time modelling approach, a stochastically parameter-varying discrete-time switched uncertain system model is developed here. To this end, define $\hat{u}_{k}:=\lim_{t\downarrow r_{k}}\hat{u}(t)$ and $\hat{y}_{k}:=\lim_{t\downarrow r_{k}}\hat{y}(t)$. $\hat{y}_{k-1}=\hat{y}(t_{k})$ and $\hat{u}_{k-1}=\hat{u}(t_{k})$ because of the left-continuity of $\hat{y}, \hat{u}$. Hence, the exact discretisation of the plant state is given by
\begin{align*}
&x_{\p}(t_{k+1})=e^{A_{\p}h_{k}}x_{\p}(t_{k})+\int^{h_{k}}_{0}e^{A_{\p}s}dsB_{\p}[u(t_{k})+e_{u}(t_{k})]\nonumber\\
&\quad +\beta_{k}\int^{h_{k}-\tau_{k}}_{0}e^{A_{\p}s}dsB_{\p}\Gamma_{u}(\sigma_{k})[\epsilon_{u}(t_{k})-e_{u}(t_{k})].
\end{align*}
Following the similar fashion, the discretisation of the controller \eqref{eqn-2} can be achieved. Define the augmented state $\bar{x}_{k}:=(x_{\p}(t_{k}), x_{\ct}(t_{k}), e_{y}(t_{k}),e_{u}(t_{k}))$, the controlled output $\bar{z}_{k}:=(y(t_{k}), u(t_{k}))$, and $\bar{\epsilon}_{k}:=(\epsilon_{y}(t_{k}), \epsilon_{u}(t_{k}))$. Therefore, the discrete-time model of the whole system is given by
\begin{align}
\label{eqn-6}
\begin{cases}
\bar{x}_{k+1}=\mathcal{A}_{\sigma_{k}, h_{k}, \tau_{k}}\bar{x}_{k}+\Upsilon\mathcal{B}_{\sigma_{k}, h_{k}, \tau_{k}}\bar{\epsilon}_{k} \\
\qquad \quad +\Upsilon_{k}\mathcal{B}_{\sigma_{k}, h_{k}, \tau_{k}}\bar{w}_{k},\\
\bar{z}_{k}=H_{\sigma_{k}}\bar{x}_{k},
\end{cases}
\end{align}
where $\bar{w}_{k}:=\bar{\epsilon}_{k}-\bar{e}_{k}$ is an auxiliary variable, $\mathcal{A}_{\sigma_{k}, h_{k}, \tau_{k}}\in\mathbb{R}^{n_{x}\times n_{x}}$, $\mathcal{B}_{\sigma_{k}, h_{k}, \tau_{k}}\in\mathbb{R}^{n_{x}\times n_{z}}$, $\mathcal{G}_{h_{k}}\in\mathbb{R}^{n_{x}\times n_{w}}$,  $H_{\sigma_{k}}\in\mathbb{R}^{n_{z}\times n_{x}}$ with $n_{x}=n_{\p}+n_{\ct}+n_{z}$, $n_{z}=n_{y}+n_{u}$ and
\begin{align*}
\mathcal{A}_{\sigma_{k}, h_{k}, \tau_{k}}&=\begin{bmatrix}\begin{smallmatrix}
A_{h_{k}}+E_{h_{k}}BDC & E_{h_{k}}BD-E_{h_{k}-\tau_{k}}B\bar{\Upsilon}\Gamma_{\sigma_{k}} \\
C(I-A_{h_{k}}-E_{h_{k}}BDC) & \mathcal{A}_{22}
\end{smallmatrix}\end{bmatrix}, \\
\mathcal{B}_{\sigma_{k}, h_{k}, \tau_{k}}&=\begin{bmatrix}
E_{h_{k}-\tau_{k}}B\Gamma_{\sigma_{k}} \\
D^{-1}\Gamma_{\sigma_{k}}-CE_{h_{k}-\tau_{k}}B\Gamma_{\sigma_{k}}
\end{bmatrix}, \\
H_{\sigma_{k}}&=\begin{bmatrix} DC & D-I \end{bmatrix}, A_{h_{k}}=\begin{bmatrix} e^{A_{\p}h_{k}} & 0 \\  0& e^{A_{\ct}h_{k}} \end{bmatrix}, \\
B&=\begin{bmatrix} 0 & B_{\p} \\ B_{\ct} & 0  \end{bmatrix}, C=\begin{bmatrix}C_{\p} & 0 \\  0& C_{\ct} \end{bmatrix},   D=\begin{bmatrix} I & 0 \\ D_{\ct} & I \end{bmatrix}, \\
\Upsilon&=\diag\{\bar{\Upsilon}, \bar{\Upsilon}\},  \quad \Upsilon_{k}=\diag\{\bar{\Upsilon}_{k}, \bar{\Upsilon}_{k}\}, \\
\Gamma_{\sigma_{k}}&=\diag\{\Gamma_{y}(\sigma_{k}), \Gamma_{u}(\sigma_{k})\}, \\
E_{\rho}&=\diag\left\{\int^{\rho}_{0}e^{A_{\p}s}ds, \int^{\rho}_{0}e^{A_{\ct}s}ds\right\},
\end{align*}
where $\mathcal{A}_{22}:=I-D^{-1}\bar{\Upsilon}\Gamma_{\sigma_{k}}+C(E_{h_{k}-\tau_{k}}B\bar{\Upsilon}\Gamma_{\sigma_{k}}-E_{h_{k}}BD)$.

\subsection{Time-scheduling Protocols}
\label{sec-protocols}

The time-scheduling protocols decide which node is granted to access to the network, and are introduced for the discrete-time system \eqref{eqn-6}. Here, two classes of the protocols are presented; see \cite{Nesic2009unified, Donkers2011stability, Donkers2012stability}.

\subsubsection{Quadratic Protocol}
For some given matrices $Q_{i}=Q^{\top}_{i}>0$, $i\in\mathfrak{L}$, with appropriate dimensions, if the time-scheduling protocol is of the form:
\begin{equation}
\label{eqn-7}
\sigma_{k}=\min\arg\max_{1\leq i\leq L}\begin{bmatrix}
\bar{x}_{k} \\
\bar{\epsilon}_{k}
\end{bmatrix}^{\top}Q_{i}\begin{bmatrix}
\bar{x}_{k} \\
\bar{\epsilon}_{k}
\end{bmatrix},
\end{equation}
then it is said to be a quadratic protocol. Based on $\bar{x}_{k}$ and $\bar{\epsilon}_{k}$, $Q_{i}$ in \eqref{eqn-7} can be rewritten as a block matrix with 4 blocks. In particular, if there exists more than one node such that the binomial \eqref{eqn-7} achieves its minimum, then the node with the smallest index is granted to access to the network. If $\sigma_{k}=\min\{\arg\max_{1\leq k\leq L}(e_{k}-\bar{\epsilon}_{k})^{\top}(e_{k}-\bar{\epsilon}_{k})\}$ with $e_{k}=(e_{y}(t_{k}), e_{u}(t_{k}))$, then the quadratic protocol is called Try-Once-Discard (TOD) protocol. In this case, the protocol depends on the network-induced error only.

\subsubsection{Periodic Protocol}
If there exists $N_{1}\in\mathbb{N}^{+}$ such that the time-scheduling protocol satisfies
\begin{equation}
\label{eqn-8}
\sigma_{k}=\sigma_{k+N_{1}}\in\mathfrak{L}, \quad  \forall k\in\mathbb{N}^{+},
\end{equation}
then it is called a periodic protocol, and $N_{1}$ is a period of the protocol. A special case of \eqref{eqn-8} is Round-Robin (RR) protocol, which is defined as the periodic protocol with $N_{1}=L$. RR protocol implies that each node has and only has one chance to access to the network in a period.

\subsection{Problem Statement}
\label{sec-stability}

The objective of this paper is to establish sufficient conditions such that, in the presence of the network and quantizer, the discrete-time system \eqref{eqn-6} is mean-square exponentially input-to-state stable and the $\mathcal{H}_{\infty}$ performance is achieved. To this end, the following definitions are presented.

\begin{definition}
\label{def:1}
The system \eqref{eqn-6} is \emph{exponentially mean-square input-to-state stable (EMSISS)}, if there exist $c_{1}, c_{2}>0$ and $\gamma_{1}>0$ such that for all initial condition $\bar{x}_{0}\in\mathbb{R}^{n_{x}}$ and a sequence of the random variables $\{(h_{k}, \tau_{k}): k\in\mathbb{N}^{+}\}$,
\begin{equation}
\label{eqn-9}
\E\{\|\bar{x}_{k}\|^{2}\}\leq c_{1}\|\bar{x}_{0}\|^{2}e^{-c_{2}k}+\gamma_{1}\sup_{i\in\{0, \ldots, k-1\}}\|\bar{\epsilon}_{i}\|^{2}.
\end{equation}
\end{definition}

\begin{definition}
\label{def:2}
The system \eqref{eqn-6} is said to achieve the \emph{$\mathcal{H}_{\infty}$ performance}, if for each sequence of the random variables $\{(h_{k}, \tau_{k}): k\in\mathbb{N}^{+}\}$,
\begin{equation}
\label{eqn-10}
\sum^{\infty}_{k=0}\E\{\|\bar{z}_{k}\|^{2}\}\leq c_{3}\|\bar{x}_{0}\|^{2}+\gamma_{2}\sum^{\infty}_{k=0}\|\bar{\epsilon}_{i}\|^{2},
\end{equation}
where $\gamma_{2}>0$ is the prescribed scalar.
\end{definition}

\section{Overapproximation of Discrete-Time System}
\label{sec-overapproximation}

Due to the existence of the stochastic uncertainties in the parameters $h_{k}$ and $\tau_{k}$, it is not easy to address the stability analysis by using the model \eqref{eqn-6}. To deal with this issue, the system \eqref{eqn-6} is further overapproximated by a polytopic system with norm-bounded uncertainties. The polytopic system is given below:
\begin{align}
\label{eqn-11}
\bar{x}_{k+1}&=\left[\sum^{N}_{n=1}\alpha_{kn}\bar{A}_{\sigma_{k}, n}+\bar{B}\Delta\bar{C}_{\sigma_{k}}\right]\bar{x}_{k} \nonumber \\
&\quad+\Upsilon\left[\sum^{N}_{n=1}\alpha_{kn}\bar{E}_{\sigma_{k}, n}+\bar{B}\Delta\bar{F}_{\sigma_{k}}\right]\bar{\epsilon}_{k}\nonumber \\
&\quad+\Upsilon_{k}\left[\sum^{N}_{n=1}\alpha_{kn}\bar{E}_{\sigma_{k}, n}+\bar{B}\Delta\bar{F}_{\sigma_{k}}\right]\bar{w}_{k},
\end{align}
where $N$ is the number of vertices of the polytopic and $\bar{A}_{\sigma, n}\in\mathbb{R}^{n_{x}\times n_{x}}$,
$\bar{B}\in\mathbb{R}^{n_{x}\times n_{1}}$, $\bar{C}_{\sigma}\in\mathbb{R}^{n_{1}\times n_{x}}$,
$\bar{E}_{\sigma, n}\in\mathbb{R}^{n_{x}\times n_{z}}$, and $\bar{F}_{\sigma}\in\mathbb{R}^{n_{1}\times n_{z}}$. For each $k\in\mathbb{N}^{+}$, the vector $\alpha_{k}=(\alpha_{k1}, \ldots, \alpha_{kN})^{\top}\in\mathscr{A}$ with $\mathscr{A}=\{\alpha\in\mathbb{R}^{N}: \sum^{N}_{n=1}\alpha_{n}=1, \alpha_{n}\geq0, n\in\mathfrak{N}:=\{1, \ldots, N\}\}$, $\Delta_{k}\in\bm{\Delta}$, where $\bm{\Delta}$ is a norm-bounded set to describe the additive uncertainty. To construct the polytopic system \eqref{eqn-11}, a procedure of overapproximating the discrete-time system \eqref{eqn-6} is presented below. This procedure is based on gridding and norm-bounding approach \cite{Sala2005computer, Skaf2009analysis, Donkers2011stability} and is a generalization of the procedure applied in \cite{Donkers2011stability, Donkers2012stability}.

\noindent\textbf{Procedure}
\begin{enumerate}
  \item Let the desired thresholds be $\varsigma>0$ and $\varpi^{\ast}>0$. Choose distinct points $(\bar{h}_{n}, \bar{\tau}_{n})\in\Theta$, $n\in\{1, \ldots, N\}$ to form a set $\mathcal{P}:=\{(\bar{h}_{n}, \bar{\tau}_{n}): n\in\{1, \ldots, N\}\}$ such that $\Theta$ is in the closure of $\mathcal{P}$. Based on the set $\mathcal{P}$, pick $M$ triangles $\mathcal{S}_{m}\in(\Theta+\delta\mathbb{B}^{\circ})$, $m\in\mathfrak{M}:=\{1, 2, \ldots, M\}$, such that
  \begin{enumerate}[(i)]
    \item for any $m_{1}, m_{2}\in\mathfrak{M}$ and $m_{1}\neq m_{2}$, $\prb\{(h, \tau)\in\mathcal{S}_{m_{1}}\cap\mathcal{S}_{m_{2}}\}=0$;
    \item $\Theta\subseteq\cup^{M}_{m=1}\mathcal{S}_{m}\subseteq(\Theta+\delta\mathbb{B}^{\circ})$;
    \item $\prb\{\cup^{M}_{m=1}\mathcal{S}_{m}\backslash\Theta\}\in[0, \varsigma]$.
  \end{enumerate}
  \item Computer the probability $\bar{p}_{m}=\prb\{(h, \tau)\in\mathcal{S}_{m}\}$. Define $\bar{A}_{\sigma n}:=\mathcal{A}_{\sigma, \bar{h}_{n}, \bar{\tau}_{n}}$, $\bar{E}_{\sigma n}:=\mathcal{B}_{\sigma, \bar{h}_{n}, \bar{\tau}_{n}}$, where $(\bar{h}_{n}, \bar{\tau}_{n})\in\mathcal{S}_{m}$ and $\sigma\in\mathfrak{L}$.
  \item To bound the approximation errors, the matrix $\bar{\Lambda}$ is constructed and decomposed to be Jordan form $\bar{\Lambda}=\diag\{A_{\p}, A_{\ct}\}=T\Lambda T^{-1}$, $\Lambda=\diag\{\Lambda_{1}, \ldots, \Lambda_{K}\}$, where $\Lambda_{i}\in\mathbb{R}^{n_{i}\times n_{i}}$ is the $i$-th Jordan block of the matrix $\bar{\Lambda}$.
  \item For each Jordan block $\Lambda_{i}$, compute the worst-case approximation errors: $\delta_{A, i}:=\max_{m\in\mathfrak{M}}\delta^{m}_{A, i}$, $\delta_{E_{h}, i}:=\max_{m\in\mathfrak{M}}\delta^{m}_{E_{h}, i}$, and $\delta_{E_{h-\tau}, i}:=\max_{m\in\mathfrak{M}}\delta^{m}_{E_{h-\tau}, i}$, with
  \begin{align*}
  \delta^{m}_{A, i}&:=\max_{\bar{\alpha}_{l}\in\mathscr{A}_{1}}\left\|e^{\Lambda_{i}\sum^{3}_{l=1}\bar{\alpha}_{l}\bar{h}_{ml}}
  -\sum^{3}_{l=1}\bar{\alpha}_{l}e^{\bar{h}_{ml}\Lambda_{i}}\right\|, \\
  \delta^{m}_{E_{h}, i}&:=\max_{\bar{\alpha}_{l}\in\mathscr{A}_{1}}\left\|\sum^{3}_{l=1}\bar{\alpha}_{l}\int^{\sum^{3}_{l=1}  \bar{\alpha}_{l}\bar{h}_{ml}}_{\bar{h}_{ml}}e^{\Lambda_{i}s}ds\right\|, \\
  \delta^{m}_{E_{h-\tau}, i}&:=\max_{\bar{\alpha}_{l}\in\mathscr{A}_{1}}\left\|\sum^{3}_{l=1}\bar{\alpha}_{l}\int^{\sum^{3}_{l=1}
  \bar{\alpha}_{l}(\bar{h}_{ml}-\bar{\tau}_{ml})}_{(\bar{h}_{ml}-\bar{\tau}_{ml})}e^{\Lambda_{i}s}ds\right\|,
  \end{align*}
  where $\mathscr{A}_{1}:=\{\alpha\in\mathbb{R}^{3}: \sum^{3}_{n=1}\alpha_{n}=1, \alpha_{n}\geq0\}$; $(\bar{h}_{ml}, \bar{\tau}_{ml})\in\mathcal{P}$ are the vertices of the triangle $\mathcal{S}_{m}$ with $m\in\{1, \ldots, M\}$ and $l\in\{1, 2, 3\}$.

  \item Define $\bar{B}:=\tilde{B}U$, $\tilde{B}:=\begin{bmatrix}\begin{smallmatrix}
  T & T & T \\
  -CT & -CT & -CT
  \end{smallmatrix}\end{bmatrix}$ and
  \begin{align*}
  U&:=\diag\{\delta_{A, 1}I_{1}, \ldots, \delta_{A, K}I_{K}, \delta_{E_{h}, 1}I_{1}, \ldots, \\
   &\quad  \delta_{E_{h}, K}I_{K}, \delta_{E_{h-\tau}, 1}I_{1}, \ldots, \delta_{E_{h-\tau}, K}I_{K}\},\\
  \bar{C}_{\sigma}&:=\begin{bmatrix}
  T^{-1} & 0 \\
  T^{-1}BDC & T^{-1}BD \\
  0 & -T^{-1}B\Gamma_{\sigma}
  \end{bmatrix},
  \bar{F}_{\sigma}:=\begin{bmatrix}  0 \\  0 \\  T^{-1}B\Gamma_{\sigma}\end{bmatrix},
  \end{align*}
  where the matrix $J_{\sigma}$ with appropriate dimension is additive and the identity matrix $I_{i}$ corresponds to the Jordan block $\Lambda_{i}$.
  \item Compute $\varpi:=\max_{\sigma\in\mathfrak{L}}\{\|\bar{B}\|\|\bar{C}_{\sigma}\|\}$. If $\varpi>\varpi^{\ast}$, then the set $\Theta$ need to be partitioned with a larger $M$. Otherwise, the tightness of the overapproximation is achieved. The additive uncertainty and the disturbance-induced uncertainty are bounded by the set:
  \begin{align*}
  \bm{\Delta}&:=\{\diag\{\Delta_{1}, \ldots, \Delta_{3K}\}: \Delta_{i+jK}\in\mathbb{R}^{n_{i}\times n_{i}}, \\
  &\quad \|\Delta_{i+j K}\|\leq1, i\in\{1, \ldots, K\}, j\in\{0, 1, 2\}\}.
  \end{align*}
\end{enumerate}

\begin{remark}
\label{rmk:6}
Our procedure of overapproximation is a generalization of but different from the one applied in \cite{Donkers2012stability}. First, two desired thresholds are specified in this paper, whereas only one threshold was given in \cite{Van2013stability, Donkers2012stability}. Although the choice of $(h_{k}, \tau_{k})$ is allowed to be more than $\Theta$, the degree of the exceeding needs to be limited. Thus, the threshold $\varsigma$ is to detect the degree of the exceeding part in the partition. Moreover, the combination of the thresholds $\varsigma$ and $\varpi^{\ast}$ is to detect the degree of the partition, whereas only $\varpi^{\ast}$ is used in \cite{Van2013stability, Donkers2012stability} to detect the tightness of the partition. Second, the coefficient matrix $J_{\sigma}$ in $\bar{H}_{\sigma}$ is added in Step 5 to scale the effects of the disturbance on each node. For every node, data transmission is impacted by the disturbance to different extent. If a node is chosen, then combining the threshold $\varpi^{\ast}$ and the coefficient matrix $J_{\sigma}$ leads to the boundedness of the disturbance-induced uncertainty and the compatibility of the norm-bounded uncertainty set $\bm{\bar{\Delta}}$. The coefficient matrix $J_{\sigma}$ can be specified \emph{a priori} according the ability of fault tolerance and the importance of each node.
\hfill $\square$
\end{remark}

Based on the above procedure, the following lemma is obtained. Its proof is a modification of the proof of Theorem \uppercase\expandafter{\romannumeral3}.2 in \cite{Donkers2011stability}, and hence omitted here.

\begin{lemma}
\label{lem-1}
The polytopic system obtained by \textbf{Procedure} is an overapproximation of the system \eqref{eqn-6}, that is,
\begin{align*}
&\left\{\begin{bmatrix}
\tilde{A}_{\sigma, h, \tau} & \tilde{B}_{\sigma, h, \tau} \\
\end{bmatrix}: (h, \tau)\in\Theta\right\}\subseteq\left\{\sum^{N}_{n=1}\alpha_{n}\begin{bmatrix}
\bar{A}_{\sigma n} & \bar{E}_{\sigma n} \\
\end{bmatrix} \right.\\
&\quad \left.+\bar{B}\Delta\begin{bmatrix}
\bar{C}_{\sigma} & \bar{F}_{\sigma} \\
\end{bmatrix}: \alpha\in\mathcal{A}, \Delta\in\bm{\Delta}, \sigma\in\mathfrak{L}\right\},
\end{align*}
where the set $\bm{\Delta}$ is given in Step 6 of \textbf{Procedure}.
\end{lemma}

Based on Lemma \uppercase\expandafter{\romannumeral2}. 4 in \cite{Donkers2011stability}, the stability of the discrete-time system \eqref{eqn-6} implies the stability of the continuous-time systems given by \eqref{eqn-1}--\eqref{eqn-2}. According to Lemma \ref{lem-1}, the stability of the polytopic system \eqref{eqn-11} leads to the stability of the discrete-time system \eqref{eqn-6}. Therefore, to guarantee the stability of the continuous-time systems given by \eqref{eqn-1}--\eqref{eqn-2}, what we need do next is to derive sufficient conditions for the stability of the polytopic system \eqref{eqn-11}.

\section{Stability and $\mathcal{H}_{\infty}$ Performance Analysis}
\label{sec-stabilityanalysis}

In this section, using the overapproximation in Section \ref{sec-overapproximation}, sufficient conditions for the stochastic stability properties of the system \eqref{eqn-11} are derived for different time-scheduling protocol and quantizer cases. To this end, the following lemma is given, which states the sufficient conditions to establish EMSISS of the discrete-time system \eqref{eqn-11}.

\begin{lemma}
\label{lem-2}
Consider the system \eqref{eqn-13} and let Assumption \ref{asn-1} hold. If there exist a Lyapunov function $V: \mathbb{R}^{n_{x}}\times\mathbb{N}\rightarrow\mathbb{R}_{+}$ and constants $a_{1}, a_{2}, a_{3}, a_{4}, a_{5}>0$ with $a_{2}\geq a_{3}\geq2a_{5}$, such that for all $\bar{x}_{k}\in\mathbb{R}^{n_{x}}$, $\bar{\epsilon}_{k}\in\mathbb{R}^{n_{z}}$, $k\in\mathbb{N}_{+}$,
\begin{align}
\label{eqn-12}
a_{1}\|\bar{x}_{k}\|^{2}\leq V(\bar{x}_{k}, k)&\leq a_{2}\|\bar{x}_{k}\|^{2},\\
\label{eqn-13}
\E\{V(\bar{x}_{k+1}, k+1)\}&\leq V(\bar{x}_{k}, k)-a_{3}\|\bar{x}_{k}\|^{2} \nonumber \\
&\quad +a_{4}\|\bar{\epsilon}_{k}\|^{2}+a_{5}\|\bar{w}_{k}\|^{2},
\end{align}
then the system \eqref{eqn-13} is EMSISS with $c_{1}=a^{-1}_{1}a_{2}$, $c_{2}=-\ln(1-a^{-1}_{2}(a_{3}-2a_{5}))$ and $\gamma_{1}=a^{-1}_{1}(a_{3}-2a_{5})^{-1}a_{2}(a_{4}+2a_{5})$.
\end{lemma}

\begin{proof}
From \eqref{eqn-12}-\eqref{eqn-13} and the triangle inequality, one has
\begin{align}
\label{eqn-14}
\E\{V(\bar{x}_{k+1}, k+1)\}&\leq(1-a^{-1}_{2}(a_{3}-2a_{5}))V(\bar{x}_{k}, k) \nonumber \\
&\quad +(a_{4}+2a_{5})\|\bar{\epsilon}_{k}\|^{2}.
\end{align}
Let $\kappa=1-a^{-1}_{2}(a_{3}-2a_{5})\in[0, 1)$. Iterating \eqref{eqn-14} from $k_{0}=0$ to $k+1$ implies that
\begin{align}
&\E\{V(\bar{x}_{k+1}, k+1)\}\nonumber\\
&\leq\kappa^{k+1}V(\bar{x}_{0}, 0)+\sum^{k}_{i=0}\kappa^{i}(a_{4}+2a_{5})\|\bar{\epsilon}_{i}\|^{2}\nonumber
\end{align}
\begin{align}
\label{eqn-15}
&\leq\kappa^{k+1}a_{2}\|\bar{x}_{0}\|^{2}+\frac{a_{4}+2a_{5}}{1-\kappa}\sup_{i\in\{1, \ldots, k\}}\|\bar{\epsilon}_{i}\|^{2}.
\end{align}
We have from \eqref{eqn-12} and \eqref{eqn-15} that
\begin{align*}
\E\{\|\bar{x}_{k+1}\|^{2}\}&\leq a^{-1}_{1}a_{2}\kappa^{k+1}\|\bar{x}_{0}\|^{2} \\
&\quad +(1-\kappa)^{-1}a^{-1}_{1}(a_{4}+2a_{5})\sup_{i\in\{0, \ldots, k\}}\|\bar{\epsilon}_{i}\|^{2},
\end{align*}
which implies that the system \eqref{eqn-13} is EMSISS.
\end{proof}

In the following, sufficient conditions for EMSISS of the system \eqref{eqn-11} are derived for different time-scheduling protocols. To begin with, some auxiliary sets are introduced. Define the class of stochastic vectors as $\mathscr{U}:=\{\vartheta\in\mathbb{R}^{M}: \sum^{M}_{i=1}\vartheta_{i}=1, \vartheta_{i}\geq0, i\in\mathfrak{M}\}$. Based on the uncertainty set $\bm{\Delta}$, define $\mathscr{R}:=\{\diag\{o_{1}I_{1}, \ldots, o_{iK+j}I_{iK+j}, \ldots, o_{3K}I_{3K}\}: o_{iK+j}>0, i\in\{0, 1, 2\}, j\in\{1, 2, \ldots, K\}\}$, where the identity matrix $I_{iK+j}$ corresponds to the real Jordan matrix $\bar{\Lambda}_{iK+j}$ in Step 3 of \textbf{Procedure} with the same size.

\subsection{The Quadratic Protocol Case}
\label{sec-quadratic}

For the quadratic protocols, the candidate Lyapunov function is chosen to be
\begin{equation}
\label{eqn-16}
V(\bar{x}_{k})=\bar{x}^{\top}_{k}P_{i}\bar{x}_{k},
\end{equation}
where $i=\min\{j\in\mathfrak{L}: (\bar{x}_{k}, \bar{\epsilon}_{k})\in\Psi_{j}\}$ with $\Psi_{j}$ defined as
\begin{align*}
\left\{\begin{bmatrix}
\bar{x}_{k} \\
\bar{\epsilon}_{k}
\end{bmatrix}\in\mathbb{R}^{n_{x}+n_{z}}:\begin{bmatrix}
\bar{x}_{k} \\
\bar{\epsilon}_{k}
\end{bmatrix}^{\top}(Q_{i}-Q_{l})\begin{bmatrix}
\bar{x}_{k} \\
\bar{\epsilon}_{k}
\end{bmatrix}\leq0, \forall l\in\mathfrak{L}\right\}.
\end{align*}

\begin{figure*}[!b]
\normalsize
\hrulefill
\begin{align}
\label{eqn-17}
\Omega_{imn}:=\begin{bmatrix}\begin{smallmatrix}
\Lambda_{i} & \star & \star& \star & \star & \star & \star  & \star & \star & \star & \star\\
-\varrho_{6m}\bar{Q}^{i}_{12} & \Pi_{i} & \star & \star & \star & \star & \star & \star & \star & \star & \star \\
0 & 0 & -a_{5}\varrho_{5m}I & \star & \star & \star & \star& \star & \star & \star & \star \\
\sqrt{\bar{p}_{m}}H_{i} & 0 & 0 & -I & \star & \star & \star & \star & \star & \star & \star\\
\sqrt{\bar{p}_{m}}P_{j}\bar{A}_{in} & \sqrt{\bar{p}_{m}}P_{j}\boldsymbol{\Upsilon}\bar{E}_{in} & 0   & 0  & -P_{j} & \star & \star & \star & \star & \star & \star \\
0 & \sqrt{\bar{p}_{m}}P_{j}\boldsymbol{\Upsilon}\bar{E}_{in} & 0 & 0 & 0  & -P_{j} & \star & \star & \star & \star & \star\\
0 & 0 & \sqrt{\bar{p}_{m}}P_{j}\boldsymbol{\Upsilon}\bar{E}_{in} & 0 & 0 & 0 & -\bar{P}_{i} & \star & \star & \star & \star\\
\sqrt{\bar{p}_{m}}\bar{C}_{i}& 0 & 0& 0& 0& 0& 0 & -\rho I & \star & \star & \star\\
0 & \sqrt{\bar{p}_{m}}\bar{F}_{i}& \sqrt{\bar{p}_{m}}\bar{F}_{i} & 0 & 0& 0 & 0& 0 & -\rho I & \star & \star  \\
0& 0 &0& 0 & \rho\bar{B}^{\top}P_{j} & 0 & 0& 0& 0 & -\rho I & \star \\
0& 0 & 0 & 0 & \rho\bar{B}^{\top}\boldsymbol{\Upsilon} P_{j} & \rho\bar{B}^{\top}\boldsymbol{\Upsilon} P_{j} & \rho\bar{B}P_{j} & 0& 0& 0 & -\rho I
\end{smallmatrix}\end{bmatrix},
\end{align}
\end{figure*}

\begin{theorem}
\label{thm-1}
Let Assumptions \ref{asn-1}-\ref{asn-3} hold and the probability distribution for $(h, \tau)$ be given. If there exist $P_{i}=P^{\top}_{i}>0$ with $i\in\mathfrak{L}$, $\varrho_{1}, \varrho_{2}, \varrho_{3}, \varrho_{4}, \varrho_{5}, \varrho_{6}\in\mathscr{U}$ and $\rho>0$ such that \eqref{eqn-17} holds for all $i\in\mathfrak{L}$, $m\in\mathfrak{M}$, $n\in\mathfrak{N}$, where $\Lambda_{i}:=-\varrho_{1m}P_{i}-\varrho_{6m}\bar{Q}^{i}_{11}+\varrho_{5m}\gamma_{3}I$, $\Pi_{i}:=-\varrho_{3m}\gamma_{1}I-\varrho_{6m}\bar{Q}^{i}_{22}$, $\bar{Q}^{i}_{11}=\sum^{L}_{l=1}\zeta_{i, l}(Q^{i}_{11}-Q^{l}_{11})$, $\bar{Q}^{i}_{12}=\bar{Q}^{i\top}_{21}=\sum^{L}_{l=1}\zeta_{i, l}(Q^{i}_{12}-Q^{l}_{12})$, $\bar{Q}^{i}_{22}=\sum^{L}_{l=1}\zeta_{i, l}(Q^{i}_{22}-Q^{l}_{22})$, $\bar{P}_{j}=\boldsymbol{\Upsilon}\odot P_{j}$, $\boldsymbol{\Upsilon}=\diag\{(1-\bar{\alpha})\bar{\alpha}I, (1-\bar{\beta})\bar{\beta}I, (1-\bar{\alpha})\bar{\alpha}I, (1-\bar{\beta})\bar{\beta}I\}$, then the system \eqref{eqn-11} with the quadratic protocol is EMSISS and the $\mathcal{H}_{\infty}$ performance is achieved.
\end{theorem}

\begin{proof}
First, define $a_{1}:=\min_{i\in\mathfrak{L}}\lambda_{\min}(P_{i})$ and $a_{2}:=\max_{i\in\mathfrak{L}}\lambda_{\max}(P_{i})$. From the Lyapunov function \eqref{eqn-16},
\begin{equation}
\label{eqn-18}
a_{1}\|\bar{x}_{k}\|^{2}\leq V(\bar{x}_{k})\leq a_{2}\|\bar{x}_{k}\|^{2}.
\end{equation}

Next is to prove the dissipative property of the Lyapunov function \eqref{eqn-16} for all $(\bar{x}_{k}, \bar{\epsilon}_{k})\in\Psi_{i}$, $\bar{x}_{k+1}\in\mathbb{R}^{n_{x}}$ and $i\in\mathfrak{L}$. To analyze the EMSISS property, it is sufficient to show that there exist $a_{3}, a_{4}, a_{5}>0$ and $\zeta_{i, l}\geq0$ such that
\begin{align}
\label{eqn-19}
&\E\{V(\bar{x}_{k+1})\}-V(\bar{x}_{k})\leq-a_{3}\|\bar{x}_{k}\|^{2}+a_{4}\|\bar{\epsilon}_{k}\|^{2} \nonumber \\
&\quad +a_{5}\|\bar{\omega}_{k}\|^{2}+\sum^{L}_{l=i}\zeta_{i, l}\begin{bmatrix}
\bar{x}_{k} \\
\bar{\epsilon}_{k}
\end{bmatrix}^{\top}(Q_{i}-Q_{l})\begin{bmatrix}
\bar{x}_{k} \\
\bar{\epsilon}_{k}
\end{bmatrix}.
\end{align}
To study the $\mathcal{H}_{\infty}$ performance, it suffices to prove (see \cite{Wang2007robust})
\begin{equation}
\label{eqn-20}
\E\{V(\bar{x}_{k+1})\}-V(\bar{x}_{k})\leq-\|\bar{z}_{k}\|^{2}+\gamma_{2}\|\bar{\epsilon}_{k}\|^{2}.
\end{equation}
Combining \eqref{eqn-19} and \eqref{eqn-20} and consider the external disturbances, we obtain that, to analyze both the EMSISS property and the $\mathcal{H}_{\infty}$ performance, it suffices to prove
\begin{align}
\label{eqn-21}
&\E\{V(\bar{x}_{k+1})\}-V(\bar{x}_{k})\leq-\|\bar{z}_{k}\|^{2}-a_{3}\|\bar{x}_{k}\|^{2}+a_{4}\|\bar{\epsilon}_{k}\|^{2} \nonumber  \\
&\quad +a_{5}\|\bar{\omega}_{k}\|^{2}+\sum^{L}_{l=i}\zeta_{i, l}\begin{bmatrix}
\bar{x}_{k} \\
\bar{\epsilon}_{k}
\end{bmatrix}^{\top}(Q_{i}-Q_{l})\begin{bmatrix}
\bar{x}_{k} \\
\bar{\epsilon}_{k}
\end{bmatrix}.
\end{align}
Therefore, $\gamma_{2}=a_{4}+a_{5}$ if the $\mathcal{H}_{\infty}$ performance is achieved.

Based on the discrete-time system \eqref{eqn-6} and the candidate Lyapunov function \eqref{eqn-16}, it holds that in the case of no disturbance,
\begin{align*}
&\E\{V(\bar{x}_{k+1})\}=\E\{\mathfrak{X}^{\top}_{k}\Psi^{\top}_{i}(h_{k}, \tau_{k})P_{j}\Psi_{i}(h_{k}, \tau_{k})\mathfrak{X}_{k}\} \\
&\leq\sum^{M}_{m=1}\bar{p}_{m}\max_{(h_{k}, \tau_{k})\in\mathcal{S}_{m}}\mathfrak{X}^{\top}_{k}\Psi^{\top}_{i}(h_{k}, \tau_{k})P_{j}\Psi_{i}(h_{k}, \tau_{k})\mathfrak{X}_{k},
\end{align*}
where $\Psi^{\top}_{i}(h_{k}, \tau_{k}):=(\mathcal{A}_{\sigma_{k}, h_{k}, \tau_{k}}, \Upsilon\mathcal{B}_{\sigma_{k}, h_{k}, \tau_{k}}, \Upsilon_{k}\mathcal{B}_{\sigma_{k}, h_{k}, \tau_{k}})$. To guarantee \eqref{eqn-21} for all $\bar{x}_{k}\in\mathbb{R}^{n_{x}}$, $\bar{\epsilon}_{k}\in\mathbb{R}^{n_{x}}$, $k\in\mathbb{N}^{+}$, it suffices to show that there exists $\zeta_{i, l}\geq0$ such that for all $(\bar{h}_{m}, \bar{\tau}_{m})\in\mathcal{S}_{m}$, $m\in\mathfrak{M}$ and all $i\in\mathfrak{L}$,
\begin{align*}
&\sum^{M}_{m=1}\bar{p}_{m}\mathfrak{X}^{\top}_{k}\Psi^{\top}_{i}(h_{k}, \tau_{k})P_{j}\Psi_{i}(h_{k}, \tau_{k})\mathfrak{X}_{k}\leq\bar{x}^{\top}_{k}P_{i}\bar{x}_{k}-\|\bar{z}_{k}\|^{2}   \\
&\quad -a_{3}\|\bar{x}_{k}\|^{2}+a_{4}\|\bar{\epsilon}_{k}\|^{2}+a_{5}\|\bar{\omega}_{k}\|^{2}  \\
&\quad  +\sum^{L}_{l=i}\zeta_{i, l}
\begin{bmatrix}
\bar{x}_{k} \\
\bar{\epsilon}_{k}
\end{bmatrix}^{\top}\begin{bmatrix}Q^{i}_{11}-Q^{l}_{11} & Q^{i}_{12}-Q^{l}_{12} \\ \star& Q^{i}_{22}-Q^{l}_{22}\end{bmatrix}\begin{bmatrix}
\bar{x}_{k} \\
\bar{\epsilon}_{k}
\end{bmatrix},
\end{align*}
which holds if
\begin{align}
\label{eqn-22}
&\bar{p}_{m}\mathfrak{X}^{\top}_{k}\Psi^{\top}_{i}(h_{k}, \tau_{k})P_{j}\Psi_{i}(h_{k}, \tau_{k})\mathfrak{X}_{k}\leq\varrho_{1m}\bar{x}^{\top}_{k}P_{i}\bar{x}_{k} \nonumber\\
&-\varrho_{2m}\|\bar{z}_{k}\|^{2}-a_{3}\varrho_{3m}\|\bar{x}_{k}\|^{2}
+a_{4}\varrho_{4m}\|\bar{\epsilon}_{k}\|^{2}+a_{5}\varrho_{5m}\|\bar{\omega}_{k}\|^{2} \nonumber \\
& +\varrho_{6m}\begin{bmatrix}
\bar{x}_{k} \\
\bar{\epsilon}_{k}
\end{bmatrix}^{\top}\begin{bmatrix}\bar{Q}^{i}_{11} & \bar{Q}^{i}_{12} \\ \star& \bar{Q}^{i}_{22}\end{bmatrix}\begin{bmatrix}
\bar{x}_{k} \\
\bar{\epsilon}_{k}
\end{bmatrix},
\end{align}
where $\varrho_{i}=(\varrho_{i1}, \ldots, \varrho_{iM})\in\mathscr{U}$ and $i\in\{1, \ldots, 6\}$. The sufficient condition for \eqref{eqn-22} is that for all $\alpha_{k}\in\mathscr{A}$, $\Delta\in\bm{\Delta}$, $i\in\mathfrak{L}$ and $m\in\mathfrak{M}$, $\mathfrak{X}^{\top}_{k}\Phi_{1}\mathfrak{X}_{k}\leq0$ holds with
\begin{align*}
\Phi_{1}&=\begin{bmatrix}
\Phi_{11} & \star & \star \\ \Phi^{\top}_{12} & \Phi_{22} & \star \\ \Phi^{\top}_{13} & \Phi^{\top}_{23} & \Phi_{33}
\end{bmatrix},
\end{align*}
where $\hat{A}_{in}:=\bar{A}_{in} +\bar{B}\Delta\bar{C}_{i}$, $\hat{E}_{in}:=\bar{E}_{in}+\bar{B}\Delta\bar{F}_{i}, \Phi_{11}:=\bar{p}_{m}(\sum^{N}_{n=1}\alpha_{kn}\hat{A}_{in})^{\top}P_{j}(\sum^{N}_{n=1}\alpha_{kn}\hat{A}_{in})-\varrho_{1m}P_{i}
+\varrho_{2m}H^{\top}_{i}H_{i}+\varrho_{3m}a_{3}I-\varrho_{6m}\bar{Q}^{i}_{11}, \Phi_{12}:=\bar{p}_{m}(\sum^{N}_{n=1}\alpha_{kn}\hat{A}_{in})^{\top}P_{j}\boldsymbol{\Upsilon}(\sum^{N}_{n=1}\alpha_{kn}\hat{E}_{in})
-\varrho_{6m}\bar{Q}^{i}_{12}, \Phi_{13}=0, \Phi_{22}:=\bar{p}_{m}(\sum^{N}_{n=1}\alpha_{kn}\hat{E}_{in})^{\top}\boldsymbol{\Upsilon}^{\top}
P_{j}\boldsymbol{\Upsilon}(\sum^{N}_{n=1}\alpha_{kn}\hat{E}_{in})-\varrho_{4m}a_{4}I-\varrho_{6m}\bar{Q}^{i}_{22}, \Phi_{23}=0, \Phi_{33}=(\sum^{N}_{n=1}\alpha_{kn}\hat{E}_{in})^{\top}\bar{P}_{j}(\sum^{N}_{n=1}\alpha_{kn}\hat{E}_{in})-a_{5}\varrho_{5m}I$ with $\bar{Q}^{i}_{11}, \bar{Q}^{i}_{12}, \bar{Q}^{i}_{22}$ are defined in the theorem.

Applying Schur complement lemma, $\Phi_{1}\leq0$ is equivalent to $\sum^{N}_{n=1}\alpha_{kn}\Phi_{2}\leq0$, where the matrix $\Phi_{2}$ is defined as
\begin{align*}
\begin{bmatrix}\begin{smallmatrix}
\Lambda_{i} & \star & \star& \star & \star & \star & \star \\
-\varrho_{6m}\bar{Q}^{i}_{12} & \Pi_{i} & \star & \star & \star & \star & \star \\
0 & 0 & -a_{5}\varrho_{5m}I & \star & \star & \star & \star \\
\sqrt{\bar{p}_{m}}H_{i} & 0 & 0 & -I & \star & \star & \star \\
\sqrt{\bar{p}_{m}}P_{j}\hat{A}_{in} & \sqrt{\bar{p}_{m}}P_{j}\boldsymbol{\Upsilon}\hat{E}_{in}  & 0   & 0  & -P_{j} & \star & \star \\
0 & \sqrt{\bar{p}_{m}}P_{j}\boldsymbol{\Upsilon}\hat{E}_{in}  & 0 & 0 & 0  & -P_{j} & \star \\
0 & 0 & \sqrt{\bar{p}_{m}}P_{j}\boldsymbol{\Upsilon}\hat{E}_{in}  & 0 & 0 & 0 & -\bar{P}_{i}
\end{smallmatrix}\end{bmatrix}.
\end{align*}
Moreover, $\Phi_{2}\leq0$ can be rewritten as $\mathcal{M}+\mathcal{N}\mathcal{O}\mathcal{P}+\mathcal{P}^{\top}\mathcal{O}^{\top}\mathcal{N}^{\top}\leq0$, where $\mathcal{O}=\diag\{\Delta, \Delta\}$ and
\begin{align*}
\mathcal{M}&=\begin{bmatrix}\begin{smallmatrix}
\Lambda_{i} & \star & \star& \star & \star & \star & \star \\
-\varrho_{6m}\bar{Q}^{i}_{12} & \Pi_{i} & \star & \star & \star & \star & \star \\
0 & 0 & -a_{5}\varrho_{5m}I & \star & \star & \star & \star \\
\sqrt{\bar{p}_{m}}H_{i} & 0 & 0 & -I & \star & \star & \star \\
\sqrt{\bar{p}_{m}}P_{j}\bar{A}_{in} & \sqrt{\bar{p}_{m}}P_{j}\boldsymbol{\Upsilon}\bar{E}_{in} & 0   & 0  & -P_{j} & \star & \star \\
0 & P_{j}\boldsymbol{\Upsilon}\sqrt{\bar{p}_{m}}\bar{E}_{in} & 0 & 0 & 0  & -P_{j} & \star \\
0 & 0 & P_{j}\boldsymbol{\Upsilon}\sqrt{\bar{p}_{m}}\bar{E}_{in} & 0 & 0 & 0 & -\bar{P}_{i}
\end{smallmatrix}\end{bmatrix}, \\
\mathcal{N}^{\top}&=\begin{bmatrix}
\sqrt{\bar{p}_{m}}\bar{C}_{i}& 0 & 0& 0& 0& 0& 0\\
0 & \sqrt{\bar{p}_{m}}\bar{F}_{i}& \sqrt{\bar{p}_{m}}\bar{F}_{i} & 0 & 0& 0 & 0
\end{bmatrix}, \\
\mathcal{P}&=\begin{bmatrix}
0& 0 &0& 0 &\bar{B}^{\top}P_{j} & 0 & 0 \\
0& 0 & 0 & 0 &\bar{B}^{\top}\boldsymbol{\Upsilon}P_{j} &\bar{B}^{\top}\boldsymbol{\Upsilon} P_{j} &\bar{B}P_{j}
\end{bmatrix}.
\end{align*}
Based on the S-procedure in \cite{Wang2007robust}, $\mathcal{M}+\mathcal{N}\mathcal{O}\mathcal{P}+\mathcal{P}^{\top}\mathcal{O}^{\top}\mathcal{N}^{\top}\leq0$ if and only if there exists $\rho>0$ such that
\begin{align*}
\begin{bmatrix}
\mathcal{M} & \star & \star \\
\mathcal{N}^{\top} & -\rho I & \star  \\
\rho\mathcal{P} & 0 & -\rho I
\end{bmatrix}\leq0,
\end{align*}
which further holds due to the LMI condition \eqref{eqn-21}. Hence, the proof is completed.
\end{proof}

\subsection{The Periodic Protocol Case}
\label{sec-periodic}

For the periodic protocol \eqref{eqn-8}, the candidate Lyapunov function is chosen to be
\begin{equation}
\label{eqn-23}
V(\bar{x}_{k})=\bar{x}^{\top}_{k}P_{i}\bar{x}_{k}.
\end{equation}
where $i=k\pmod L$. As the counterpart of Theorem \ref{thm-2}, the following theorems are presented. Their proofs are along the similar lines of the proofs of Theorem \ref{thm-2}, and omitted here.

\begin{theorem}
\label{thm-2}
Let Assumptions \ref{asn-1}-\ref{asn-3} hold and the probability distribution for $(h, \tau)$ be given. If there exist $P_{i}=P^{\top}_{i}>0$ with $P_{\mathfrak{a}L+i}=P_{i}$ for all $\mathfrak{a}\in\mathbb{N}$ and $i\in\mathfrak{L}$, $\varrho_{1}, \varrho_{2}, \varrho_{3}, \varrho_{4}, \varrho_{5}\in\mathscr{U}$ and $\rho>0$ such that \eqref{eqn-17} holds for all $i\in\mathfrak{L}$, $m\in\mathfrak{M}$, $n\in\mathfrak{N}$, where $\Lambda_{i}:=-\varrho_{1m}P_{i}+\varrho_{5m}\gamma_{3}I$, $\Pi_{i}:=-\varrho_{3m}\gamma_{1}I$, $\bar{Q}^{i}_{12}=\bar{Q}^{i\top}_{21}=0$, $P_{j}=P_{i+1}$, $\bar{P}_{i}=\boldsymbol{\Upsilon}\odot P_{i}$, $\boldsymbol{\Upsilon}=\diag\{(1-\bar{\alpha})\bar{\alpha}I, (1-\bar{\beta})\bar{\beta}I, (1-\bar{\alpha})\bar{\alpha}I, (1-\bar{\beta})\bar{\beta}I\}$, then the system \eqref{eqn-11} with the periodic protocol is EMSISS and the $\mathcal{H}_{\infty}$ performance is achieved.
\end{theorem}

Due to the convergence and boundedness of the system state and Assumption \ref{asn-3}, the quantization error $\bar{\epsilon}\in\mathbb{R}^{n_{z}}$ is bounded and convergent along the time line.
In terms of the EMSISS, we have that
\begin{align}
\label{eqn-24}
\limsup_{k\rightarrow\infty}\E\{\|\bar{x}_{k}\|^{2}\}&\leq\gamma_{1}\sup_{k\rightarrow\infty}\|\bar{\epsilon}_{k}\|^{2} \nonumber \\
&\leq\gamma_{1}\sup_{k\rightarrow\infty}\|\Lambda\mu_{k}\|^{2},
\end{align}
where $\mu_{k}$ is the quantization parameter at the time instance $t_{k}$. Observe from \eqref{eqn-24} that there exists an ultimate bound for the system state. Since the quantization parameter is non-increasing, $\sup_{k\rightarrow\infty}\|\mu(t_{k})\|^{2}=\|\mu(t_{0})\|^{2}$, and thus the ultimate bound is finite and only related to the initial value of the quantization parameter. In addition, with the convergence of the quantization error, the system state is convergent asymptotically and thus $\lim_{k\rightarrow\infty}\E\{\|\bar{x}_{k}\|^{2}\}=0$. In terms of the $\mathcal{H}_{\infty}$ performance analysis, the zero initial condition, which is need in many existing works \cite{Yang2006control, Wang2007robust}, leads to $\bar{\epsilon}\equiv0$, and thus is not required here. Instead, $\bar{\epsilon}$ is treated as a vanishing disturbance for the system \eqref{eqn-10}.

\section{Numerical Example}
\label{sec-example}

In this section, a benchmark example is borrowed from \cite{Heemels2010networked, Heemels2009networked} to illustrate the obtained results. The system dynamics and the applied continuous-time controller are given in the form of \eqref{eqn-1}-\eqref{eqn-2} with the matrices specified as follows:
\begin{align*}
A_{\p}&=\begin{bmatrix}1.380 & -0.208 & 6.715& -5.676\\ -0.581 &4.290 & 0& 0.675 \\ 1.067 &4.273 &-6.654 & 5.893 \\ 0.048 &4.273 & 1.343 &-2.104 \end{bmatrix}, \\
B_{\p}&=\begin{bmatrix} 0 & 0 \\ 5.679 & 0 \\ 1.136 & -3.146\\1.136 & 0 \end{bmatrix}, \quad C_{\p}=\begin{bmatrix}1 & 0 & 1 & -1 \\ 0 & 1 & 0 & 0 \end{bmatrix}, \\
A_{\ct}&=0,   B_{\ct}=\begin{bmatrix} 0 & 1  \\ 1 & 0  \end{bmatrix},
C_{\ct}=\begin{bmatrix} -2 & 0  \\ 0 & 8  \end{bmatrix}, D_{\ct}=\begin{bmatrix} 0 & -2 \\ 5 & 0 \end{bmatrix}.
\end{align*}
To evaluate the MATI $h_{\mati}$ and the MAD $h_{\mad}$, we assume here that only the plant output is transmitted over the network, whereas the control input is transmitted to the actuators directly; see \cite{Heemels2010networked}. Therefore, there are $L=2$ sensor nodes for $y_{1}$ and $y_{2}$, which further implies that $\Gamma_{1}=\diag\{1, 0\}$ and $\Gamma_{2}=\diag\{0, 1\}$. The lower bounds on the transmission intervals and delays are set the same as $10^{-3}$. Both the MATI $h_{\mati}$ and the MAD $h_{\mad}$ needed to be established to guarantee the stability and performance of NQCS with communication delays.

\begin{figure}[!t]
\begin{center}
\begin{picture}(60, 100)
\put(-60,-15){\resizebox{60mm}{40mm}{\includegraphics[width=2.5in]{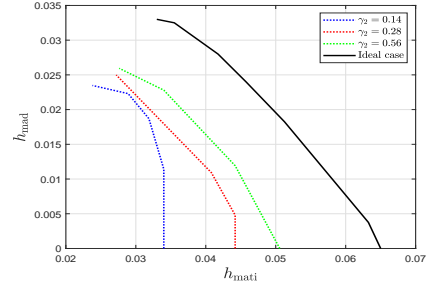}}}
\end{picture}
\end{center}
\caption{Illustration of the relations among the $\mathcal{H}_{\infty}$ attenuation level $\gamma_{2}$, the MATI $h_{\mati}$ and the MAD $h_{\mad}$ in the RR protocol case.}
\label{fig-2}
\end{figure}

Here, we consider the TOD protocol. First, we assume that $\operatorname{Pr}((h, \tau)\in\Theta)=1,$ where $\Theta=\{(h, \tau)\in\mathbb{R}^{2}: h\in[10^{-3}, 0.1], \tau\in[10^{-3}, 0.05]\}$. Furthermore, we let $\operatorname{Pr}((h, \tau)\in\Theta)=\int_{\Theta}p(h, \tau)dhd\tau$, where $p(h, \tau)$ is the corresponding marginal probability density function (MPDF). More specifically, we consider the uniform MPDF given by $p(h, \tau)=(0.1-10^{-3})^{-1}(0.05-10^{-3})^{-1}$ for $(h, \tau)\in\Theta$ and $h\geq\tau$, and $p(h, \tau)=0$ elsewhere.

To verify the stability and the $\mathcal{H}_{\infty}$ performance, the first step is to develop the approximated model \eqref{eqn-11} by implementing the proposed procedure. To simply the partition of the set $\Theta$, we can partition the intervals $[10^{-3}, 0.1]$ and $[10^{-3}, 0.05]$ into $N_{a}$ and $N_{b}$ uniform parts, respectively. Therefore, we have $2N_{a}N_{b}$ triangles and $(N_{a}+1)(N_{b}+1)$ vertices. By increasing $N_{a}$ and $N_{b}$, the approximation accuracy is improved, whereas the computational complexity is getting higher. After the partition, we can compute the probability in each triangle, and obtain all the matrices $\bar{A}_{\sigma n}, \bar{E}_{\sigma n}, \bar{B}, \bar{C}_{\sigma}$ and $\bar{F}_{\sigma}$ with $n\in\{1, \ldots, (N_{a}+1)(N_{b}+1)\}$ and $\sigma\in\mathfrak{L}$. In addition, we assume that the plant output is quantized by a zoom quantizer with the constraints $\max_{j\in\{1, 2\}}\{\Lambda_{j}\}=0.8$ and $\max_{j\in\{1, 2\}}\{\Omega_{j}\}=0.6$. Using the NCS toolbox \cite{Bauer2012networked}, the feasibility of the obtained LMIs in Theorem \ref{thm-1} and the minimization of the $\mathcal{H}_{\infty}$ attenuation level $\gamma_{2}$ are verified by obtaining admissible $h_{\mati}$ and $h_{\mad}$ for the desired system performance. Therefore, the upper bound on $\gamma_{2}$ is related to both $h_{\mati}$ and $h_{\mad}$. Let $a_{3}=10^{-2}$ and $a_{5}=10^{-4}$. Given the dropout rates with $(\bar{\alpha}, \bar{\beta})=(0.8, 1)$, the tradeoff curves for different $\mathcal{H}_{\infty}$ attenuation levels are depicted in Fig. \ref{fig-2}. In Fig. \ref{fig-2}, the ideal case means the deterministic case without quantization and dropouts, and corresponds to the case of the infinity $\mathcal{H}_{\infty}$ attenuation level. Therefore, this case provides the boundaries for the upper bounds of $h_{\mati}$ and $h_{\mad}$. We observe from Fig. \ref{fig-2} that $h_{\mati}$ and $h_{\mad}$ is getting non-conservative with the increase of $\gamma_{2}$. In addition, $h_{\mati}$ and $h_{\mad}$ are affected by the dropout rates, that is, the increase of the dropout rates results in the decrease of $h_{\mati}$ and $h_{\mad}$.

\section{Conclusion}
\label{sec-conclusion}

This paper addressed stability and $\mathcal{H}_{\infty}$ performance analysis problem of stochastic networked and quantized control systems. Using the discrete-time modelling approach, the plant and the controller were discretised and a stochastically discrete-time system model was obtained. To facilitate the stability analysis, the discrete-time system was overapproximated by a polytopic system with additional norm-bounded uncertainties. For the polytopic system with different protocols and quantizers, sufficient conditions were established for mean-square exponential input-to-state stability and the $\mathcal{H}_{\infty}$ performance. In the future, our work will focus on the controller and observer design of stochastic networked control systems and networked control systems with large transmission delays.


\end{document}